\documentclass[12pt]{article}

\usepackage{enumerate,amsmath,amsfonts,amsthm, amssymb, bm,slashed}

\newcommand{\R}{\mathbb{R}}

\usepackage{tensor}

\newtheorem{lemma}{Lemma}

\newtheorem*{remark}{Remark}

\usepackage{fullpage}

\title{\Huge{Kinetic dynamics of neutral spin particles\\ in a spacetime with torsion}}

 \author
 {
       Simone Calogero  \\
       {\small Department of Mathematical Sciences}  \\
       {\small Chalmers University of Technology}  \\
       {\small Gothenburg, Sweden} \\
       }

\date{}

\begin{document}
\maketitle
\begin{abstract}
A kinetic model for the dynamics of collisionless spin neutral particles  in a spacetime with torsion is proposed. The fundamental matter field is the kinetic density $f(x,u,s)$ of particles with four-velocity $u$ and four-spin $s$.  The stress-energy tensor and the spin current of the particles distribution are defined as suitable integral moments of $f$ in the $(u,s)$ variables. By requiring compatibility with the contracted Bianchi identity in Einstein-Cartan theory, we derive a transport equation on the kinetic density $f$ that generalizes the well-known Vlasov equation for spinless particles. The total number of particles in the new model is not conserved. To restore this important property we assume the existence in spacetime of a second species of particles with the same mass and spin magnitude. The Vlasov equation on the kinetic density $\overline{f}$ of the new particles is derived by requiring that the sum of total numbers of particles of the two species should be conserved. 
\end{abstract}

\section{Introduction}
Matter in general relativity is represented by tensor fields on a four-dimensional manifold $M$---the spacetime. To measure distances, as well as fields strength, the manifold $M$ is equipped with a metric $g$. For consistency with special relativity, $g$ is assumed to have Lorentzian signature. To write down field equations for the metric $g$ and the matter fields one also needs a connection on $M$. In Einstein theory this is chosen to be the Levi-Civita connection, i.e., the unique torsion-free connection $\nabla$ such that $\nabla g =0$. With this choice the Ricci tensor $R_{ab}$ is symmetric, and therefore the Einstein equation
\begin{equation}\label{Einsteineq}
R_{ab}-\frac{1}{2}g_{ab}R=8\pi T_{ab}\quad (G=c=1)
\end{equation}
requires the stress-energy(-momentum) tensor $T_{ab}$ of the matter to be symmetric as well.

A few years after the publication of general relativity, Cartan proposed to modify the geometry of spacetime by removing the restriction that the connection should be torsion-free~\cite{Cartan1, Cartan2}. One consequence of this generalization is that the Ricci tensor is no longer necessarily symmetric, and so neither is the stress-energy tensor of the matter (assuming that the Einstein equations still hold in the form~\eqref{Einsteineq} when the connection has a torsion).

Shortly after Cartan's original work, 
Pauli~\cite{Pauli} proposed that the electron possessed an additional degree of freedom (besides position and velocity), which was later identified by  Goudsmit and Uhlenbech as a form of intrinsic angular momentum---the electron spin~\cite{GU}. In the following years several theories were put forward attempting to describe the electron spin in classical mechanics, some of which are presented in the review~\cite{Nyborg}. Cartan, among others, suggested that spin could act as a source of spacetime torsion, thereby establishing a link between microscopic physics and general relativity. This idea was 
largely forgotten once the concept of spin found a successful approach in quantum mechanics (and later quantum field theory), until the articles~\cite{Kibble, Sciama} and the popular review~\cite{HHK} contributed to its revival. Einstein-Cartan's theory is nowadays an active (but not prevailing) research topic in the physics community. 
As spin effects become important only at exceedingly small scales, and only in the interior of matter, the current experimental evidence of general relativity does not rule out Einstein-Cartan's theory as an alternative description of gravity~\cite{Hammond,Trautman}.   

The purpose of this paper is to lay down the foundations of a new kinetic model for spin particles in a spacetime with torsion. 
Neglecting spin, the fundamental matter field in (general relativistic) kinetic theory is a non-negative function $f=f(x,u)$ representing the number density of particles at the point $x\in M$ with four-velocity $u\in T_xM$~\cite{Hakan,Ehlers, Alan}. In this paper we consider an extended particle density $f=f(x,u,s)$, depending also on the particle spin four-vector $s\in T_xM$.  
The stress-energy tensor $T_{ab}(x)$ and the spin current $\tensor{S}{_a_b^c}(x)$ of the particles distribution are defined as suitable integral moments of $f$ in the $(u,s)$ variables. The stress-energy tensor appears as a source of curvature in the Einstein equation; the spin current $\tensor{S}{_a_b^c}$ generates a torsion $\tensor{C}{_a_b^c}$ in spacetime, which we assume to obey 
Cartan's equation
\[
\tensor{C}{_a_b^c}=8\pi(2\tensor{S}{_a_b^c}+\tensor{\delta}{_a^c}\tensor{S}{_b_d^d}-\tensor{\delta}{_b^c}\tensor{S}{_a_d^d}).
\]
As the spin current in our model  satisfies the Frenkel condition $\tensor{S}{_a_b^b}=0$, Cartan's equation simplifies to
\begin{equation}\label{cartaneq}
\tensor{C}{_a_b^c}=16\pi\tensor{S}{_a_b^c}.
\end{equation}
The metric connection with torsion $\tensor{C}{_a_b^c}$ is denoted by $\widehat{\nabla}$.
The Einstein equation for the spacetime metric is, under Frenkel's condition,
\begin{equation}\label{newEeq} 
\widehat{R}_{ab}-\frac{1}{2}g_{ab}\widehat{R}=8\pi \Sigma_{ab},\quad \Sigma_{ab}=T_{ab}-\widehat{\nabla}_c(\tensor{S}{_a_b^c}+2\tensor{S}{^c_{(ab)}}),
\end{equation}
where 
$\widehat{R}_{ab}$ is the Ricci tensor of the connection $\widehat{\nabla}$ and $\widehat{R}=\tensor{\widehat{R}}{^a_a}$. 
The system~\eqref{cartaneq}-\eqref{newEeq} is a special case of the system of field equations derived in~\cite{Kibble, Sciama} by the variational principle; for this reason,~\eqref{cartaneq}-\eqref{newEeq} are also known as Kibble-Sciama equations in the physics literature.

The main difficulty in developing kinetic matter models in Einstein-Cartan's theory is to find an evolution equation on the particle density $f$ that is consistent with the Bianchi identity 
for the connection $\widehat{\nabla}$, which is~\cite{Penrose}
\begin{equation}\label{BianchiID}
\widehat{\nabla}_b\left(\tensor{\widehat{R}}{^a^b}-\frac{1}{2}g^{ab}\widehat{R}\right)=-\tensor{C}{^a^b^c}\tensor{\widehat{R}}{_c_b}+\frac{1}{2}\tensor{C}{^b^c^d}\tensor{\widehat{R}}{^a_d_b_c}.
\end{equation} 
Combining~\eqref{cartaneq}--\eqref{BianchiID} we obtain the conservation law of energy-momentum in the form
\begin{equation}\label{consThat}
\widehat{\nabla}_b\Sigma^{ab}=-2\tensor{S}{^a^b^c}\tensor{\widehat{R}}{_c_b}+\tensor{S}{^b^c^d}\tensor{\widehat{R}}{^a_d_b_c}.
\end{equation}
As opposed to other matter models (e.g., spin fluids~\cite{Medina}), in kinetic theory the constraint~\eqref{consThat} does not automatically entail evolution equations on the matter fields; on the contrary, it complicates the task of deriving admissible kinetic models. In this paper we introduce an evolution equation on the kinetic density $f(x,u,s)$ of spin particles which is compatible with~\eqref{consThat} and which generalizes the well-known Vlasov model for spinless particles~\cite{Hakan, Alan}. The overall interpretation of Vlasov equations is that the kinetic density $f$ remains constant along the trajectories of the individual particles; in particular, Vlasov models neglect collisions among the particles.

A questionable feature of our model is that, when applied to a single species of particles, it violates the conservation law of the total number of particles. To restore this important property, we assume the existence in spacetime of a second species of particles with the same mass and spin magnitude of the first species; the evolution equation for the kinetic density $\overline{f}(x,u,s)$  of these new particles is obtained by imposing that the total particles number current computed with the kinetic density $f+\overline{f}$ should be divergence-free. 

The rest of the paper is organized as follows.
In Section~\ref{EVsec} we recall some facts on the classical kinetic theory for spinless particles and the Einstein-Vlasov system. 
In Section~\ref{ECsec} we review the results on Einstein-Cartan's theory used in this paper, including the derivation of the system~\eqref{cartaneq}-\eqref{newEeq}. In Section~\ref{spinkinsec} we introduce the definition of kinetic density for particles with spin and show how to construct the tensors $T_{ab}$,  $\tensor{S}{_a_b^c}$ from it. Section~\ref{VlasovSec} is dedicated to the special case of collisionless spin particles. In Section~\ref{Vlasovant} we extend our model to the case when two species of particles with the same mass and spin magnitude are present. In Section~\ref{conclu} we summarize our results and comment further on their physical meaning.

In this paper we assume that particles are neutral.
Charged spin particles will of course also generate a torsion in spacetime, but the dominant effect of spin in the charged case is to induce a magnetic moment on the particles. The generalization of the model in this paper to the case of charged particles will be discussed elsewhere.  

{\bf Notation.} We shall often (but not always) use the abstract index notation to denote tensors and tensor operations~\cite{Wald}. The Latin letters $a,b,c,d,e,h$ denote abstract indexes. Greek letters denote component indexes and run from 0 to 3; the spatial component indexes are denoted by $i,j,k,\dots$ and run from 1 to 3. Single indexes within brackets, as in $e_{(\mu)}$, label the vectors of a basis and are not component indexes.  
The metric has signature $(-,+,+,+)$ and physical units are chosen so that $
G=c=\hbar=1$.


\section{Kinetic theory for spinless particles}\label{EVsec}
Let $(M,g)$ be a spacetime---i.e., a four-dimensional time-oriented Lorentzian manifold. For the moment we do not make any specific choice for a connection on $M$. Consider a matter distribution that consists of a large number of identical point particles with rest mass $m> 0$. In this section we assume that the state of each particle is determined by the spacetime position $x\in M$ and the four-velocity $u\in T_x M$ of the particle; in particular, the particles spin is neglected. The particles four-velocity is constrained by the condition $g_{ab}(x)u^au^b=-1$, which has to be satisfied at every point $x\in M$. The state space of each particle is therefore the seven-dimensional submanifold of the tangent bundle given by
\[
\Pi=\cup_{x\in M}\Pi[x],\quad \Pi[x] = \{u\in T_xM:g_{ab}(x)u^au^b=-1,\ u\ \text{future directed}\}.
\]
In kinetic theory the state of the matter distribution is described by a function 
\[
f:\Pi\to [0,\infty)
\] 
giving the number density of particles in state space.
\begin{remark}
\textnormal{The kinetic density $f$ is more commonly defined in the literature as a function of $(x,p)$, where $p=mu$ is the four-momentum of the particles; see~\cite{Hakan,Ehlers, Alan}. Using $(x,u)$ as independent variables has several advantages in this paper; e.g., it results in the kinetic density having the same physical dimension for particles with and without spin (in our units) and 
avoids the appearance of several constants $m$ in the equations.}
\end{remark} 
Let $d\pi(x)$ denote the volume element induced on $\Pi[x]$ by the metric volume form on $T_xM$ (where we regard $T_xM$ as a flat manifold with metric $g_{ab}(x)$ and $\Pi[x]$ as a submanifold of $T_xM$). All matter fields in spacetime are obtained by integrating momentum fields defined on $\Pi[x]$ (microscopic fields) with respect to the measure $f(x,u)\,d\pi(x)$.  For instance, the timelike vector field on $M$ given by
\begin{equation}\label{N}
(N_f)^a(x)=\int_{\Pi[x]} u^a f \,d\pi(x)
\end{equation}
represents the particles number current in spacetime. The stress-energy tensor (i.e., the momentum current) of the particles distribution is 
\begin{equation}\label{TVlasov}
(T_f)^{ab}(x)=m\int_{\Pi[x]} u^a u^bf\,d\pi(x).
\end{equation}
$T_f$ is symmetric and satisfies the strong and dominant energy conditions~\cite{Ehlers}. The Einstein equation in units $G=c=1$  reads
\begin{equation}\label{Einsteineq2}
R_{ab}-\frac{1}{2}R\,g_{ab}=8\pi (T_f)_{ab}.
\end{equation}
The Ricci tensor $R_{ab}$ depends on the connection and is not, in general, symmetric. However, as the stress-energy tensor is symmetric for kinetic matter, so must be $R_{ab}$ in the left hand side of~\eqref{Einsteineq2}. The latter holds in particular when the connection is the Levi-Civita one, which we assume to be the case in the rest of this section.

The evolution equation satisfied by the kinetic particle density $f$ depends on the type of interaction between the particles. In general it takes the form $L(f)=0$, where $L$ is an operator acting on the state space $\Pi$. If collisions among the particles are neglected, the operator $L$ is the geodesics spray and the resulting equation on $f$ is called Vlasov (or collisionless Boltzmann) equation.\footnote{The name ``Vlasov equation'' referred originally to the plasma physics version of the model~\cite{Vlasov}, but it has by now become common to adopt the same name for the model applied to gravitational systems.}  
To derive an explicit form for the Vlasov equation, we introduce a local system of coordinates $x^\mu$ in a neighborhood of $x\in M$ and an orthonormal basis $e_{(\mu)}$ of $T_xM$ such that $e_{(0)}$ is timelike and future pointing.
Let $u^\mu$ denote the components of $u\in T_xM$ in the basis $e_{(\mu)}$. We use $(x^\mu, u^\nu)$ as local coordinates on the tangent bundle and $(X_{(\mu)},U_{(\nu)})$ as basis for the tangent space of the bundle, where  $(X_{(\mu)},U_{(\nu)})\psi=(\partial_{x^\mu}\psi, \partial_{u^\nu}\psi)$ for all smooth functions $\psi:TM\to\R$. The geodesic spray in these coordinates is the vector field on the tangent bundle given by
\[
L=u^\nu e_{(\nu)}^\mu X_{(\mu)}+\tensor{\gamma}{_\alpha_\beta^\mu}u^\alpha u^\beta U_{(\mu)},
\]
where $e_{(\nu)}^\mu=e_{(\nu)}(x^\mu)$ are the components of $e_{(\nu)}$ in the basis $X_{(\mu)}$ and $\tensor{\gamma}{_\alpha_\beta^\mu}$ are the Ricci rotation coefficients of the frame $e_{(\mu)}$; that is,
\begin{equation}\label{riccicoeff}
\tensor{\gamma}{_\alpha_\beta^\mu}= e_{(\alpha)}^a e_{(\beta)}^b\nabla_ae^{(\mu)}_b,
\end{equation}
where $e^{(\mu)}_a$ is the co-frame dual to $e_{(\mu)}^a$. 
The state space conditions in the frame $e_{(\mu)}$ become
\[
\eta_{\mu\nu}u^\mu u^\nu=-1,\quad u^0>0,
\]
and thus entail 
\begin{equation}\label{p0}
u^0=\sqrt{1+|{\bm u}|^2}=-u_0,
\end{equation}  
where
\[
{\bm u}=(u^1, u^2, u^3)=(u_1,u_2,u_3),\quad |{\bm u}|^2=(u^1)^2+(u^2)^2+(u^3)^2.
\]
The kinetic particle density can be written as a function of $(x,{\bm u})$, which we denote by $f_*$; that is, 
\[
f_*(x,{\bm u})=f(x,\sqrt{1+|{\bm u}|^2},{\bm u}).
\]
Using $\gamma_{\alpha\beta\mu}=-\gamma_{\alpha\mu\beta}$, we find that the Vlasov equation $L(f_*)=0$ is 
\begin{equation}\label{XVlasov}
u^\mu e_{(\mu)}^\nu\partial_{x^\nu}f_*+\tensor{\gamma}{_\mu_\nu^i}u^\mu u^\nu \partial_{u^i}f_* =0,
\end{equation}
where it is understood that $u^0$ is given by~\eqref{p0}. 
The invariant volume element $d\pi(x)$ on $\Pi[x]$ in the coordinates $(u^1,u^2,u^3)$ is given by $d\pi(x)=\sqrt{|\det \mathfrak{h}|}\,d{\bm u}$, where $d{\bm u} = du^1\wedge du^2\wedge du^3$ and $\mathfrak{h}$ is the (Riemannian) metric induced by $\eta_{\mu\nu}$ on the hyperboloid $u^0=\sqrt{1+|{\bm u}|^2}$, i.e., 
\begin{equation}\label{hypmet}
\mathfrak{h}_{ij}=\delta_{ij}-\frac{u_iu_j}{(u^0)^2}.
\end{equation}
It follows that
\[
d\pi(x) = \sqrt{|\det\mathfrak h|}d{\bm u}=\frac{d{\bm u}}{u^0},
\]
and so the tensor fields~\eqref{N}-\eqref{TVlasov} take the form
\begin{subequations}\label{NTnew}
\begin{equation}
(N_f)^a=(N_f)^\mu e_{(\mu)}^a,\quad (T_f)^{ab}=(T_f)^{\mu\nu}e_{(\mu)}^a e_{(\nu)}^b,
\end{equation}
where
\begin{equation}
(N_f)^\mu(x) = \int_{\R^3} f_*(x,{\bm u}) u^\mu \frac{d{\bm u}}{u^0},\quad (T_f)^{\mu\nu}(x) =m\int_{\R^3} f_*(x,{\bm u}) u^\mu u^\nu \frac{d{\bm u}}{u^0}.
\end{equation}
\end{subequations}
The Vlasov equation~\eqref{XVlasov} implies that these fields are divergence-free:
\begin{equation}\label{consLaws}
\nabla_a (N_f)^a=0,\quad \nabla_b (T_f)^{ab}=0,
\end{equation}
where, for every vector field $V^a$ and tensor field $P^{ab}$, the divergence in the orthonormal frame $e_{(\mu)}$ is computed by the formulas
\begin{subequations}\label{Divergence}
\begin{align}
&\nabla_a V^a=e_{(\mu)}^\alpha\partial_{x^\alpha}V^\mu+\tensor{\gamma}{^\alpha_\alpha_\mu} V^\mu,\\ 
&\nabla_b P^{ab}=(e_{(\nu)}^\alpha\partial_{x^\alpha}P^{\mu\nu}+\tensor{\gamma}{_\alpha^\mu_\beta}P^{\beta\alpha}+\tensor{\gamma}{^\alpha_\alpha_\beta}P^{\mu\beta})e_{(\mu)}^a. 
\end{align}
\end{subequations}
\begin{remark}
\textnormal{With our definition~\eqref{riccicoeff} of the Ricci rotation coefficients, which is the same as in~\cite{Wald}, the formulas for the covariant derivative of tensor fields in an orthonormal frame are obtained from those in a holonomic basis by the formal substitutions $\partial_{x^\mu}\to e_{(\mu)}^\alpha\partial_{x^\alpha}$ and $\tensor{\Gamma}{^\mu_\alpha_\beta}\to \tensor{\gamma}{_\alpha^\mu_\beta}$, where $\tensor{\Gamma}{^\mu_\alpha_\beta}$ are the Christoffel symbols. For instance, for all vector fields $V$ we have $\tensor{\nabla_\mu V}{^\nu}=\partial_{x^\mu}V^\nu+\tensor{\Gamma}{^\nu_\mu_\alpha}V^\alpha$ in the coordinates basis $X_{(\mu)}$ and  $\tensor{\nabla_\mu V}{^\nu}=e_{(\mu)}^\alpha\partial_{x^\alpha}{V}^\nu+\tensor{\gamma}{_\mu^\nu_\alpha}V^\alpha$ in the orthonormal basis $e_{(\mu)}$.}
\end{remark}
As for all matter models in general relativity, the equation $\nabla_b (T_f)^{ab}=0$ represents the conservation law of energy-momentum in differential form; in light of the contracted Bianchi identity $\nabla_a(R^{ab}-g^{ab}R/2)=0$,
it must be satisfied in order for the Einstein equation~\eqref{Einsteineq2} to admit solutions. Similarly, the equation $\nabla_a (N_f)^a=0$ is the conservation law of the particles number in differential form. Since for particles with charge $q$ the electric four-current is given by $q(N_f)^a$, then the identity $\nabla_a (N_f)^a=0$ in the Einstein-Vlasov model is equivalent to the conservation of charge, and it is therefore required if one wants to add the Maxwell equations to the system; see~\cite{EVM} for an introduction to the Einstein-Vlasov-Maxwell system. 

In the presence of an additional species of particles with kinetic density $\overline{f}$ satisfying the Vlasov equation, the total stress-energy tensor and the total number current are given respectively by
\[
T_f+T_{\overline{f}}=T_{f+\overline{f}},\quad N_f+N_{\overline{f}}=N_{f+\overline{f}}
\]
and they are divergence-free as a consequence of each of the tensors $T_f,T_{\overline{f}},N_{f},N_{\overline{f}}$ being divergence-free.


\section{Introduction to Einstein-Cartan's theory}\label{ECsec}
\subsection{Manifolds with torsion}\label{torsec}
The purpose of this section is to recall a few important definitions and formulas in Riemann-Cartan's geometry. More details can be found e.g. in~\cite{HHK, Medina, Schouten, Trautman}. 

Let $\widehat{\nabla}$ be a metric connection on the manifold $M$.
The torsion $\tensor{C}{_a_b^c}$ of $\widehat{\nabla}$ is defined by requiring
\[
\tensor{C}{_a_b^c}X^aY^b=(\widehat{\nabla}_X Y)^c-(\widehat{\nabla}_Y X)^c-[X,Y]^c
\]
for all smooth vector fields $X,Y$. By expressing the commutator $[X,Y]$ in terms of the Levi-Civita connection $\nabla$ as
\begin{equation}\label{commutator}
[X,Y]=(\nabla_XY)-(\nabla_YX),
\end{equation}
we can rewrite the definition of torsion as
\[
\tensor{C}{_a_b^c}=\tensor{K}{_a_b^c}-\tensor{K}{_b_a^c},
\]
where  the contorsion $\tensor{K}{_a_b^c}$ of $\widehat{\nabla}$ is defined by
\[
\tensor{K}{_a_b^c}X^aY^b=(\widehat{\nabla}_XY)^c-(\nabla_XY)^c.
\]
The torsion and contorsion satisfy the properties
\begin{equation}\label{propC}
(i)\ \tensor{C}{_a_b^c}=-\tensor{C}{_b_a^c},\ (ii)\ K_{abc}=-K_{acb},\
(iii)\ \tensor{K}{_a_b^c}=\frac{1}{2}\tensor{C}{_a_b^c}+\tensor{C}{^c_{(ab)}}.
\end{equation}
Property $(i)$ is obvious, while $(ii)$ and $(iii)$ follow by the metric compatibility condition of $\widehat{\nabla}$ and $\nabla$. By $(iii)$ we obtain the additional identities 
\[
\tensor{K}{_a_b^a}=\tensor{C}{_a_b^a},\quad \tensor{K}{_{[a}_{b]}^c}=\frac{1}{2}\tensor{C}{_a_b^c},\quad \tensor{K}{_{(a}_{b)}^c}=\tensor{C}{^c_{(a}_{b)}}.
\]
The Ricci rotation coefficients of any orthonormal basis $e^a_{(\mu)}$ in the connections $\widehat{\nabla},\nabla$ are related by
\[
\tensor{\widehat{\gamma}}{_\mu_\nu^\alpha}=\tensor{\gamma}{_\mu_\nu^\alpha}-\tensor{K}{_\mu_\nu^\alpha},
\]
hence~\eqref{Divergence} imply the following identities:
\begin{subequations}\label{Divergence2}
\begin{align}
&\widehat{\nabla}_a V^a =\nabla_a V^a+\tensor{K}{_a_b^a}V^b,\\
&\widehat{\nabla}_b P^{ab}=\nabla_b P^{ab}+\tensor{K}{_b_c^a}P^{cb}+\tensor{K}{_b_ c^ b}P^{ac},
\end{align}
\end{subequations}
for all tensor fields $V^a$, $P^{ab}$. More generally, there holds
\begin{equation}\label{generalCov}
\widehat{\nabla}_a\tensor{X}{^{b_1\cdots b_k}_{c_1\dots c_l}}=\nabla_a\tensor{X}{^{b_1\cdots b_k}_{c_1\dots c_l}}+\sum_j \tensor{K}{_a_d^{b_j}}\tensor{X}{^{b_1\cdots b_{j-1}d\cdots b_k}}_{c_1\dots c_l}-\sum_j \tensor{K}{_a_{c_j}^d}\tensor{X}{^{b_1\dots b_k}_{c_1\cdots c_{j-1}d\cdots c_l}}.
\end{equation}
Let $\tensor{\widehat{R}}{_a_b_c^d}$ denote the Riemann tensor of the connection $\widehat{\nabla}$, that is
\[
\tensor{\widehat{R}}{_a_b_c^d}X^aY^bZ^c=(\widehat{\nabla}_X(\widehat{\nabla}_Y Z))^d-(\widehat{\nabla}_Y(\widehat{\nabla}_X Z))^d- (\widehat{\nabla}_{[X,Y]} Z)^d
\]
for all smooth vector fields $X,Y,Z$.
Using $(\widehat{\nabla}_XY)^c=(\nabla_XY)^c+\tensor{K}{_a_b^c}X^aY^b$, Leibniz's rule and~\eqref{commutator} we find 
\[
\tensor{\widehat{R}}{_a_b_c^d}=\tensor{R}{_a_b_c^d}+\widehat{\nabla}_a\tensor{K}{_b_c^d}-\widehat{\nabla}_b \tensor{K}{_a_c^d}+\tensor{K}{_e_c^d}(\tensor{K}{_a_b^e}-\tensor{K}{_b_a^e})+\tensor{K}{_a_c^e}\tensor{K}{_b_e^d}-\tensor{K}{_b_c^e}\tensor{K}{_a_e^d},
\]
where $\tensor{R}{_a_b_c^d}$ is the Riemann tensor of the Levi-Civita connection $\nabla$. 
From here we see that $\tensor{\widehat{R}}{_a_b_c^d}$ has the following symmetry properties in common with $\tensor{R}{_a_b_c^d}$:
\begin{equation}\label{symmetriesRbar}
\tensor{\widehat{R}}{_a_b_c^d}=-\tensor{\widehat{R}}{_b_a_c^d},\quad \tensor{\widehat{R}}{_a_b_c_d}=-\tensor{\widehat{R}}{_a_b_d_c}.
\end{equation}
However, the remaining fundamental symmetries of $\tensor{R}{_a_b_c^d}$, namely
\[
\tensor{R}{_a_b_c^d}+\tensor{R}{_c_a_b^d}+\tensor{R}{_b_c_a^d}=0, \quad
\tensor{R}{_a_b_c_d}=\tensor{R}{_c_d_a_b},
\]
are in general no longer satisfied by $\tensor{\widehat{R}}{_a_b_c^d}$.  The symmetries~\eqref{symmetriesRbar} imply that there exists only one non-zero independent contraction of $\tensor{\widehat{R}}{_a_b_c^d}$. Following the convention in~\cite{Wald}, we define the Ricci tensor by contracting the second and fourth index of the Riemann tensor, that is,
\begin{equation}\label{riccitorsion}
\widehat{R}_{ab}=\tensor{\widehat{R}}{_a_c_b^c}=R_{ab}+\widehat{\nabla}_a\tensor{K}{_c_b^c}-\widehat{\nabla}_c
\tensor{K}{_a_b^c}+\tensor{K}{_a_b^d}\tensor{K}{_c_d^c}-\tensor{K}{_c_a^d}\tensor{K}{_d_b^c},
\end{equation}
where $R_{ab}$ is the Ricci tensor of the Levi-Civita connection $\nabla$.\footnote{In~\cite{HHK}, the Ricci tensor is defined as $\widehat{R}_{ab}=\tensor{\widehat{R}}{_c_a_b^c}$ and thus differs from ours by a sign.}
Moreover, taking the trace of~\eqref{riccitorsion} we find
\begin{equation}\label{ricciscalar}
\widehat{R}=R+2\widehat{\nabla}_a \tensor{K}{_b^a^b}+\tensor{K}{_a^a^c}\tensor{K}{_b_c^b}+\tensor{K}{_c^a^b}\tensor{K}{_a_b^c}.
\end{equation}

The last property of manifolds with torsion that we need is the (contracted) Bianchi identity satisfied by the Einstein tensor $
 \widehat{E}_{ab}=\widehat{R}_{ab}-g_{ab}\widehat{R}/2$; 
 recall that $\nabla_b \tensor{E}{^a^b}=0$ holds in the Levi-Civita connection. To derive the analogous equation for $\widehat{E}_{ab}$ we start from the differential Bianchi identity for the Riemann tensor $\tensor{\widehat{R}}{_a_b_c^d}$, which reads
\[
\widehat{\nabla}_{[a} \tensor{\widehat{R}}{_b_{c]}_d^e}+\tensor{C}{_{[a}_b^h}\tensor{\widehat{R}}{_{c]}_h_d^e}=0,
\]
see~\cite[Eq.~(4.2.43)]{Penrose}. Contracting $a$ with $e$ and $b$ with $d$ we arrive to the following identity:
\begin{equation}\label{bianchitorsion}
\widehat{\nabla}_b\tensor{\widehat{E}}{^a^b}=-\tensor{C}{^a^b^c}\tensor{\widehat{R}}{_c_b}+\frac{1}{2}\tensor{C}{^b^c^d}\tensor{\widehat{R}}{^a_d_b_c}.
\end{equation}
An alternative way to obtain~\eqref{bianchitorsion} is by using~\eqref{riccitorsion}-\eqref{ricciscalar} to write $\widehat{E}_{ab}$ in terms of $E_{ab}$, and then applying~\eqref{propC},~\eqref{generalCov}, as well as $\nabla_aE^{ab}=0$, in the resulting expression. We emphasize, in particular, that the identities~\eqref{bianchitorsion} and $\nabla_aE^{ab}=0$ are equivalent.

\subsection{Einstein-Cartan equations}\label{lageq}
In this section we review the argument in~\cite{HHK, Kibble, Sciama}  to derive the field equations~\eqref{cartaneq}-\eqref{newEeq} in Einstein-Cartan's theory.
The starting point is the Lagrangian density
\[
\mathcal{L}(g,K,\Psi)=\frac{1}{16\pi}\mathcal{L}_G(g,K)+\mathcal{L}_M(g,\Psi,\widehat{\nabla}\Psi),
\]
where
\[
\mathcal{L}_G(g,K)=\widehat{R}\sqrt{|g|}\quad (|g|=-\det g)
\]
is the gravitational term, and $\mathcal{L}_M(g,\Psi,\widehat{\nabla}\Psi)$ is the matter term, depending on the matter field $\Psi$ and its covariant derivative $\widehat{\nabla}\Psi$ (and thus on the contorsion $K$). As shown in~\cite[Appendix]{HHK}, the variation of $\mathcal{L}_G$ with the respect to the metric is, up to a divergence term,\footnote{We have a different sign in~\eqref{varg} than in~\cite{HHK} due to our convention on the definition of the Ricci tensor.}
\begin{equation}\label{varg}
\frac{1}{\sqrt{|g|}}\frac{\delta\mathcal{L}_G}{\delta g_{ab}}=-\widehat{E}^{ab}-\nabla^*_c\left(\tensor{P}{^a^b^c}-\tensor{P}{^b^c^a}+\tensor{P}{^c^a^b}\right),
\end{equation}
where the operator $\nabla^*$ and the tensor $\tensor{P}{_a_b^c}$ are defined by
\[
\nabla^*_a=\widehat{\nabla}_a+\tensor{C}{_a_b^b},\quad  \tensor{P}{_a_b^c}=\frac{1}{2}\left(\tensor{C}{_a_b^c}+\tensor{\delta}{_a^c}\,\tensor{C}{_b_d^d}-\tensor{\delta}{_b^c}\,\tensor{C}{_a_d^d}\right).
\]
Hence, the Lagrangian field equation $\delta\mathcal{L}/\delta g_{ab}=0$ reads
\begin{equation}\label{Eeq}
\widehat{E}^{ab}+\nabla^*_c\left(\tensor{P}{^a^b^c}-\tensor{P}{^b^c^a}+\tensor{P}{^c^a^b}\right)=8\pi T^{ab},\quad T^{ab}=\frac{2}{\sqrt{|g|}}\frac{\delta\mathcal{L}_M}{\delta g_{ab}}.
\end{equation}
Similarly, it is found that
\[
\frac{1}{\sqrt{|g|}}\frac{\delta\mathcal{L}_G}{\delta \tensor{K}{_a_b^c}}=-2\tensor{P}{_c^b^a},
\]
and so the Lagrangian field equation $\delta\mathcal{L}/\delta  \tensor{K}{_a_b^c}=0$ is
\begin{equation}\label{Ceq}
\tensor{P}{_c^b^a}=8\pi \tensor{S}{_c^b^a},\quad \tensor{S}{_c^b^a}=\frac{1}{\sqrt{|g|}}\frac{\delta\mathcal{L}_M}{\delta \tensor{K}{_a_b^c}}.
\end{equation}
The tensor $\tensor{S}{_a_b^c}$ is the spin current. As this suffices for the applications to kinetic theory discussed later, we continue this section assuming that the spin current satisfies the Frenkel condition 
\begin{equation}\label{frenkelcond}
\tensor{S}{_a_b^a}=0.
\end{equation}
Assuming~\eqref{frenkelcond} we obtain 
\[
\tensor{P}{_a_b^a}=0,\quad\tensor{P}{_a_b^c}=\frac{1}{2}\tensor{C}{_a_b^c},\quad
\nabla^*_c\left(\tensor{P}{^a^b^c}-\tensor{P}{^b^c^a}+\tensor{P}{^c^a^b}\right)=\widehat{\nabla}_c \left(\frac{1}{2}C^{abc}+C^{c(ab)}\right),
\]
and therefore~\eqref{Eeq} and~\eqref{Ceq} simplify to
\begin{align}
&\widehat{E}_{ab}=8\pi \Sigma_{ab},\quad \Sigma_{ab}=T_{ab}-\widehat{\nabla}_c \left(\tensor{S}{_a_b^c}+2\tensor{S}{^c_{(ab)}}\right),\label{primofield}\\
&\tensor{C}{_a_b^c}=16\pi \tensor{S}{_a_b^c}.\label{cart}
\end{align}
Using the Bianchi identity~\eqref{bianchitorsion} in~\eqref{primofield}-\eqref{cart} we find
\begin{equation}\label{conservationSigma}
\widehat{\nabla}_b\Sigma^{ab}=-2\tensor{S}{^a^b^c}\tensor{\widehat{R}}{_c_b}+\tensor{S}{^b^c^d}\tensor{\widehat{R}}{^a_d_b_c}.
\end{equation}
The identitiy~\eqref{conservationSigma} is the local conservation law of energy-momentum in Einstein-Cartan's theory.  

\subsection{General relativistic form of the Einstein-Cartan equations}
Let
\begin{subequations}\label{Eeqtor}
\begin{equation}\label{H}
H_{ab}
=8\pi\left[\left(\tensor{S}{_c_a^d}+2\tensor{S}{^d_{(ca)}}\right)\left(\tensor{S}{_d_b^c}+2\tensor{S}{^c_{(db)}}\right)+\frac{1}{2}g_{ab}\left(S_{cde}+2S_{e(cd)}\right)\left(S^{dec}+2S^{c(de)}\right)\right].
\end{equation}
By Cartan's equation~\eqref{cart}, the identity $(iii)$ in~\eqref{propC} and the formula~\eqref{riccitorsion} for $\widehat{R}_{ab}$ (with $\tensor{K}{_a_b^a}=\tensor{C}{_a_b^a}=0$), we can rewrite~\eqref{primofield} in the general relativistic form
\begin{equation}\label{EinsteinWithH}
R_{ab}-\frac{1}{2}g_{ab}R=8\pi (T_{ab}+ H_{ab}),
\end{equation}
\end{subequations}
in which the spin current appears as an external matter field with stress-energy tensor $H_{ab}$. By~\eqref{Eeqtor}  and the contracted Bianchi identity $\nabla_a(R^{ab}-g^{ab}R/2)=0$, the conservation law~\eqref{conservationSigma} is equivalent to
\begin{equation}\label{conservationT}
\nabla_b T^{ab}=J^a,\quad J^a:=-\nabla_b H^{ab}.
\end{equation}
For our purpose, it is preferable to use the conservation law of energy-momentum in the form~\eqref{conservationT}, rather than~\eqref{conservationSigma}, and therefore in the following we shall write the Einstein equation for the metric in the form~\eqref{Eeqtor} instead of~\eqref{primofield}. We emphasize, however, that it is the original geometric interpretation of $\tensor{S}{_a_b^c}$ as the spacetime torsion (up to a constant) that justifies the definition of the tensor $H_{ab}$; moreover, $H_{ab}$ is not an actual stress-energy tensor and so there is no physical reason to require that it should satisfy the energy conditions commonly imposed on $T_{ab}$ in general relativity.
\begin{remark}
\textnormal{The tensor $H_{ab}$ simplifies if the spin current, or equivalently the torsion tensor, is assumed to satisfy further algebraic properties (besides the Frenkel condition). For instance, a rather common case study is that of a totally antisymmetric torsion~\cite{DR, Fabbri}; that is, $C_{a(bc)}=0$, and thus $S_{a(bc)}=0$. In this case, $S$ is the dual of a vector field $V$, i.e., $S_{abc}=\varepsilon_{abcd}V^d$, and the tensor $H_{ab}$ simplifies to 
\[
H_{ab}=
-16\pi\left(V_aV_b+\frac{1}{2}g_{ab}V^cV_c\right). 
\]}
\end{remark}


\section{Kinetic theory of spin particles}\label{spinkinsec}
In the rest of the paper we shall complete the Einstein-Cartan theory described in the previous section by constructing the matter model
within the formalism of kinetic theory for spin particles. The fundamental matter field in this theory is the kinetic density $f(x,u,s)$ of particles with four-velocity $u$ and four-spin $s$. 
In contrast to the Lagrangian approach presented in Section~\ref{lageq}, neither the matter field equations nor the stress-energy tensor $T_{ab}$ will be derived by a variational principle. Instead we rely on the relation between kinetic theory and (relativistic) particles mechanics to justify the definitions of the tensor fields $T_{ab}$, $\tensor{S}{_a_b^c}$. Subsequently, we derive an evolution equation on the kinetic density $f$ for collisionless particles with spin by an argument similar to the one presented in Section~\ref{EVsec} and requiring compatibility with the conservation law~\eqref{conservationT}. Finally, in Section~\ref{Vlasovant} we introduce a similar model in the presence of two species of particles with the same mass and spin magnitude by imposing that the their total number should be conserved.  
\subsection{Kinetic density of spin particles}\label{kineticEC}
There exist two common representations for the spin of a particle in relativistic mechanics: as a spacelike four-vector $s^a$ satisfying the constraints $s^a u_a=0$ and $
s^a s_a = \sigma^2$, or as a skew-symmetric tensor $\phi^{ab}$ satisfying $\phi_{ab}u^b=0$ and $\phi_{ab}\phi^{ab}=\sigma^2/2$, where $\sigma$ is a given positive constant. The first representation is due to Thomas~\cite{Thomas}, while the second one was introduced, independently and almost at the same time, by Frenkel~\cite{Frenkel}. The two representations are equivalent, as the variables $s^a$, $\phi^{ab}$ are related by the identities
\[
\phi_{ab}=(s_{[a}u_{b]})^\star,\quad s_a=2u^b(\phi_{ab})^\star, 
\]
where $y^\star$ denotes the dual of $y$.
In this paper we shall employ both representations of the spin variable. Specifically, the Thomas four-vector $s^a$ will appear as independent variable in the kinetic particle density, while the Frenkel tensor $\phi_{ab}$ will be used for the definition of spin current.

To formalize the definition of kinetic particle density we introduce the vector bundle 
\[
Q=\cup_{x\in M}(T_xM)^2,
\] 
whose elements we henceforth identify with the triples $(x, u, s)$, where $x\in M$ and $u,s\in T_xM$. The vector bundle $Q$ is also a smooth, twelve-dimensional manifold. The state space of spin particles is the nine-dimensional submanifold of $Q$ given by
\begin{align*}
&\Pi_{\sigma}=\cup_{x\in M}\Pi_{\sigma}[x]\subset Q,\\
&\Pi_{\sigma}[x]=\{(u,s)\in (T_xM)^2 : g_{ab}u^a u^b = -1,\, g_{ab}s^a s^b = \sigma^2,\, g_{ab}s^a u^b =0,\, u\, \text{future directed}\}.
\end{align*}
The kinetic particle density in Thomas representation is therefore a function
\begin{equation}\label{pdens}
f:\Pi_{\sigma}\to [0,\infty).
\end{equation}
To make this construction more explicit we shall now introduce a specific set of coordinates on the state space $\Pi_{\sigma}$. Let $e_{(\mu)}$ be an orthonormal frame on the tangent bundle, with $e_{(0)}$ being timelike. Let $x^\mu$ be a local system of coordinates on $M$ and $u^\mu$, $s^\mu$ denote the components in the frame $e_{(\mu)}$ of the four-vectors $u, s$; $(x^\alpha,u^\mu, s^\nu)$ are local coordinates on $Q$.
We write 
\[
{\bm s}=(s^1,s^2,s^3)=(s_1,s_2,s_3),\quad {\bm u}=(u^1,u^2,u^3)=(u_1,u_2,u_3),
\] 
and follow the notation introduced in Section~\ref{EVsec}. The state space conditions on $(T_xM)^2$ in the frame $e_{(\mu)}$ read
\begin{equation}\label{shellnormal}
\eta_{\mu\nu}u^\mu u^\nu=-1,\ u^0>0,\  \eta_{\mu\nu}u^\mu s^\nu=0,\ \eta_{\mu\nu}s^\mu s^\nu=\sigma^2.
\end{equation}
By $\eta_{\mu\nu}u^\mu u^\nu=-1$ and $u^0>0$, we have 
\begin{equation}\label{p02}
u^0=\sqrt{1+|{\bm u}|^2}=-u_0.
\end{equation} 
From $\eta_{\mu\nu}s^\mu u^\nu=0$ we infer ${\bm s}\cdot {\bm u}=s^0u^0$, where ${\bm a}\cdot{\bm b}$ denotes the standard Euclidean product of the three-dimensional vectors ${\bm a}, {\bm b}$. It follows that
\begin{equation}\label{s0}
s^0=\frac{{\bm s}\cdot {\bm u}}{u^0}=-s_0.
\end{equation}
Next we observe that 
\begin{equation}\label{obs}
\eta_{\mu\nu}s^\mu s^\nu = \mathfrak{h}_{ij}s^i s^j,
\end{equation}
where $\mathfrak{h}_{ij}$ are the components of the hyperbolic metric~\eqref{hypmet}.
The matrix $(\mathfrak{h}_{ij})$ is positive definite. Let $\sqrt{\mathfrak{h}}$ denote the square root of $\mathfrak{h}$, that is,
\begin{equation}\label{sqrth}
(\sqrt{\mathfrak{h}})_{ij}=\delta_{ij} -\frac{u_i u_j}{u^0(1+u^0)},\quad [(\sqrt{\mathfrak{h}})^{-1}]_{ij}=\delta_{ij} +\frac{u_i u_j}{1+u^0},
\end{equation}
where $(\sqrt{\mathfrak{h}})^{-1}$ is the inverse of $\sqrt{\mathfrak{h}}$. By~\eqref{obs}
we may introduce the unit vector ${\bm \omega}\in S^2$ as 
\begin{equation}\label{n}
\omega^j=\frac{1}{\sigma}(\sqrt{\mathfrak{h}})_i^{j} s^i,\quad\text{i.e.},\quad {\bm \omega}= \frac{1}{\sigma}\left({\bm s}-\frac{s^0}{1+u^0}{\bm u}\right).
\end{equation}
By~\eqref{n},  the quantity $\sigma{\bm \omega}$ is the spin vector in the rest frame of the particle. Moreover, by~\eqref{s0} and~\eqref{n}, 
\begin{equation}\label{sofn}
{\bm s}=\sigma\left({\bm \omega} + \frac{{\bm \omega}\cdot{\bm u}}{1+u^0}\,{\bm u}\right),\quad s^0=\sigma {\bm \omega}\cdot{\bm u}.
\end{equation}
It follows that $\Pi_{\sigma}[x]\simeq \R^3\times S^2$ and so the particle density~\eqref{pdens} can be written as a function of $(x,{\bm u},{\bm \omega})$; see~\eqref{newf} below. By further introducing the angles $\theta,\varphi$ through
\begin{equation}\label{thetaphi}
{\bm \omega}=(\sin\theta\cos\varphi,\sin\theta\sin\varphi,\cos\theta),\quad (\theta,\varphi)\in [0,\pi]\times [0,2\pi),
\end{equation}
then $f$ becomes a function of $(x,{\bm u},\theta,\varphi)$. However, except for a few calculations where it is convenient to do so, we shall not employ the spherical coordinates $(\theta,\varphi)$.
The orthonormal components of the Frenkel tensor $\phi_{ab}$ in the coordinates $({\bm u},{\bm \omega})$ are given by
\begin{equation}\label{phimunu}
\phi_{\mu\nu}=\frac{1}{2}\varepsilon_{\mu\nu\alpha\beta}s^\alpha u^\beta \Rightarrow\left\{\begin{array}{l}(\phi_{01},\phi_{02},\phi_{03})=\displaystyle{\frac{\sigma}{2}}  ({\bm \omega}\wedge {\bm u}),\\[0.5cm](\phi_{23},\phi_{31},\phi_{12})=-\displaystyle{\frac{\sigma}{2}} u^0(\sqrt{\mathfrak{h}}){\bm \omega},\\[0.5cm] \,\phi_{\mu\nu}=-\phi_{\nu\mu},\end{array}\right.
\end{equation}
where $\varepsilon_{\mu\nu\alpha\beta}$ denotes the four-dimensional Levi-Civita permutation symbol.

\subsection{Spacetime matter fields}\label{matterfieldssec}
To define the stress-energy tensor and the spin current in spacetime we need first to introduce a volume form on $\Pi_{\sigma}[x]$, which we do as follows. Consider the natural metric on $(T_xM)^2$:
\begin{equation}\label{Gmetric}
\mathfrak{G}((u,s),(u_*,s_*))=g(u,u_*)+ g(s,s_*).
\end{equation}
Since $\mathfrak{G}(y,y)=\sigma^2-1$ holds for $y=(u,s)\in\Pi_{\sigma}[x]$, then for $\sigma^2=1$ the metric induced on $\Pi_{\sigma}[x]$ by $\mathfrak{G}$ is degenerate and thus its volume form is singular. Therefore, 
\[
 \text{{\it in the rest of the paper we assume that} $\sigma\neq  1$}.
 \] 
 Let $\mathfrak{H}_{\sigma}$ be the metric induced by $\mathfrak{G}$ on $\Pi_{\sigma}[x]$ for $\sigma\neq  1$ and let $d\pi_{\sigma}(x)$ be its metric volume form. 
The particles number current $(N_f)^a$ and the stress-energy tensor $(T_f)^{ab}$ for particles with spin are given by
\begin{equation}\label{NT}
(N_f)^a(x) = \int_{\Pi_{\sigma}[x]}u^af\, d\pi_{\sigma}(x),\quad (T_f)^{ab}(x) = m\int_{\Pi_{\sigma}[x]}  u^a u^bf\, d\pi_{\sigma}(x).
\end{equation}
Due to the interpretation of $f$ as the particles number density on state space and of $u$ as the particles four-velocity, these are the only physically reasonable definitions of the fields $(N_f)^a$ and $(T_f)^{ab}$. Likewise, in agreement with the general definition of spacetime currents in kinetic theory, the spin current should be given by the integral over $\Pi_{\sigma}[x]$ in the measure $f d\pi_{\sigma}(x)$ of a microscopic field of the form $\kappa_{ab} u^c$, where $\kappa_{ab}=\kappa_{ab}(u,s)$ is a skew-symmetric tensor representing the spin variable of the individual particles. Choosing $\kappa_{ab}$ to be the Frenkel tensor $\phi_{ab}$ 
leads us to define the spin current in spacetime as 
\begin{equation}\label{Sabc}
\tensor{(S_f)}{_a_b^c}(x) =  \int_{\Pi_{\sigma}[x]}  \phi_{ab}  u^c f\, d\pi_{\sigma}(x).
\end{equation}
In particular, $\tensor{(S_f)}{_a_b^c}$ satisfies $\tensor{(S_f)}{_a_b^c}=-\tensor{(S_f)}{_b_a^c}$ and, using $\phi_{ab}u^b=0$, the Frenkel condition holds:
\begin{equation}
\tensor{(S_f)}{_a_b^a}=0.
\end{equation}

The definitions of the tensor felds $N_f,T_f, S_f$ 
will now be written in the coordinates system $({\bm u}, {\bm \omega})$ on $\Pi_{\sigma}[x]$ introduced in Section~\ref{kineticEC}. We begin by computing the volume form $d\pi_{\sigma}(x)$ in these coordinates. We claim that
\begin{subequations}\label{volform}
\begin{equation}
d\pi_{\sigma}(x)=\sigma^2\sqrt{|1-\sigma^2|}\,\frac{d{\bm u}\,d{\bm \omega}}{u^0},
\end{equation}
where $d{\bm u} = du^1 \wedge du^2\wedge du^3$ and $d{\bm \omega}$ is the standard surface element on the unit sphere; that is,
employing the spherical coordinates representation~\eqref{thetaphi} of ${\bm \omega}\in S^2$, 
\begin{equation}
d{\bm \omega}=\sin\theta\, d\theta\wedge d\varphi.
\end{equation}
\end{subequations}
To prove~\eqref{volform}, let $y=(u,s)=(\sqrt{1+|{\bm u}|^2}, {\bm u}, s^0, {\bm s})$ and $z=({\bm u},\theta,\varphi)$.  Denote by $\mathcal{H}$ the $5\times 5$ matrix of components of $\mathfrak{H}_{\sigma}$ and by $\mathcal{G}$ the $8\times 8$ matrix of components of $\mathfrak{G}$. Let $\partial_zy$ be the $5\times 8$ matrix
\[
\partial_zy=\left(\begin{array}{cc}\partial_{{\bm u}}u & \partial_{{\bm u}}s\\ \partial_{(\theta,\varphi)}u & \partial_{(\theta,\varphi)}s\end{array}\right)=\left(\begin{array}{cc}\partial_{{\bm u}}u & \partial_{{\bm u}}s\\ 0 & \partial_{(\theta,\varphi)}s\end{array}\right).
\]
Then
\[
\mathcal{H}=(\partial_z y)\mathcal{G}(\partial_zy)^T=\left(\begin{array}{cc} \mathcal{A} & \mathcal{B}\\ \mathcal{B}^T & \mathcal{C}\end{array}\right),
\]
where $\mathcal{A},\mathcal{B},\mathcal{C}$ are the the matrices
\begin{align*}
&\mathcal{A}=(\partial_{{\bm u}}u)\eta(\partial_{{\bm u}}u)^T+(\partial_{{\bm u}}s)\eta(\partial_{{\bm u}}s)^T,\\
&\mathcal{B}=(\partial_{{\bm u}}s)\eta(\partial_{(\theta,\varphi)}s),\\
&\mathcal{C}=(\partial_{(\theta,\varphi)}s)\eta(\partial_{(\theta,\varphi)}s)=\sigma^2\left(\begin{array}{ll} 1 & 0\\ 0 & \sin^2\theta\end{array}\right).
\end{align*}
Therefore, 
\[
\det (\mathcal{H})=\det(\mathcal{C})\det(\mathcal{A}-\mathcal{B}\mathcal{C}^{-1}\mathcal{B}^T)=\sigma^4\sin^2\theta\det(\mathcal{D}),\quad \mathcal{D}=\mathcal{A}-\frac{\mathcal{B}\mathcal{B}^T}{\sigma^2}.
\]
The computation of $\det(\mathcal{D})$ is very long and so we present only the result,\footnote{Part of this computation has been carried out with Mathematica.} which is
\[
\det(\mathcal{D})=\frac{1-\sigma^2}{(u^0)^2};
\]
hence, $d\pi_{\sigma}(x)=\sqrt{|\det(\mathcal{H})|}\,d{\bm u}\, d\theta\wedge d\varphi$ is given by~\eqref{volform}.

It follows that the tensors~\eqref{NT}-\eqref{Sabc} in the coordinates $({\bm u}, {\bm \omega})$ read 
\begin{subequations} \label{NTS}
\begin{align}
&(N_f)^a=(N_f)^\mu e_{(\mu)}^a,\quad (N_f)^\mu =\int_{\R^3\times S^2}u^\mu f_{\sigma}\,\frac{d{\bm u}\,d{\bm \omega}}{u^0},\\
&(T_f)^{ab}=T^{\mu\nu}e_{(\mu)}^a e_{(\nu)}^b,\quad (T_f)^{\mu\nu}(x) = m\int_{\R^3\times S^2}u^\mu u^\nu f_{\sigma}\,\frac{d{\bm u}\,d{\bm \omega}}{u^0},\\
&\tensor{(S_f)}{_a_b^c}=\tensor{(S_f)}{_\mu_\nu^\alpha}e^{(\mu)}_ae^{(\nu)}_b e_{(\alpha)}^c,\quad \tensor{(S_f)}{_\mu_\nu^\alpha}(x) = \int_{\R^3\times S^2}\phi_{\mu\nu}u^\alpha f_{\sigma}\,\frac{d{\bm u}\,d{\bm \omega}}{u^0},
\end{align}
\end{subequations}
where $\phi_{\mu\nu}$ are given by~\eqref{phimunu} and
\begin{equation}\label{newf}
f_{\sigma}(x,{\bm u}, {\bm \omega})=\sigma^2\sqrt{|1-\sigma^2|}\,f\! \left(x,\sqrt{1+|{\bm u}|^2}, {\bm u},\sigma\, {\bm \omega}\cdot{\bm u},\sigma\, {\bm \omega} + \frac{\sigma {\bm \omega}\cdot{\bm u}}{1+\sqrt{1+|{\bm u}|^2}}{\bm u}\right).
\end{equation}
By~\eqref{NTS}, $f_{\sigma}$ is the kinetic particle density in the variables $({\bm u},{\bm \omega})\in\R^3\times S^2$. 
\begin{remark}
\textnormal{Upon introducing the spin average of the kinetic particle density $f_{\sigma}$ as
\[
f_*(x,{\bm u})=\int_{S^2}f_{\sigma}(x,{\bm u}, {\bm \omega})\,d{\bm \omega},
\]
the tensor fields $N_f$, $T_f$ in~\eqref{NTS} reduce to the particles number current and the stress-energy tensor~\eqref{NTnew} of spinless particles. Therefore, in the absence of torsion the particles spin averages out and gives no contribution to the spacetime geometry.}
\end{remark}
If only a species of particles with kinetic density $f$ is present, the Cartan equation~\eqref{cart} for the spacetime torsion is
\begin{subequations}\label{EinsteinCartanKinetic}
\begin{equation}\label{cartanKinetic}
\tensor{C}{_a_b^c}=16\pi \tensor{(S_f)}{_a_b^c}
\end{equation}
and the Einstein equation~\eqref{EinsteinWithH} for the metric is
\begin{equation}\label{einsteinKinetic}
R_{ab}-\frac{1}{2}g_{ab}R=8\pi ((T_f)_{ab}+ (H_f)_{ab}),
\end{equation}
\end{subequations}
where $(H_f)_{ab}$ is given as in~\eqref{H} with $S\equiv S_f$. 
In particular, the conservation law of energy-momentum~\eqref{conservationT} takes the form
\begin{equation}\label{consTKinetic}
\nabla_b (T_f)^{ab}=(J_f)^a,\quad (J_f)^a:=-\nabla_b (H_f)^{ab}.
\end{equation}
Suppose now that an additional species of spin particles with the same mass $m$ and spin magnitude $\sigma$ is present and let $\overline{f}$ be the kinetic density of the new particles. The total stress-energy tensor, particles number current and spin current are given by
\[
T_f+T_{\overline{f}}=T_{f+\overline{f}},\quad N_f+N_{\overline{f}}=N_{f+\overline{f}},\quad S_f+S_{\overline{f}}=S_{f+\overline{f}}
\] 
and therefore the Einstein-Cartan equations in this case are the same as~\eqref{EinsteinCartanKinetic} with the substitutions 
\[
S_f\to S_{f+\overline{f}},\quad T_f\to T_{f+\overline{f}},\quad H_f\to H_{f+\overline{f}}. 
\]
The conservation law of energy-momentum~\eqref{conservationT} in the presence of two species of particles reads
\begin{equation}\label{consmom2species}
\nabla_b [(T_f)^{ab}+(T_{\overline{f}})^{ab}]=\nabla_b (T_{f+\overline{f}})^{ab}=(J_{f+\overline{f}})^a.
\end{equation}
A simple but important observation at this point is that $J_{f+\overline{f}}\neq J_{f}+J_{\overline{f}}$ and therefore requiring~\eqref{consTKinetic} to hold for each species of particles does not imply that~\eqref{consmom2species} is satisfied. Instead, some kind interaction between the particles of one species with the particles of the other species is necessary. 


\section{The Vlasov equation for spin particles}\label{VlasovSec}
The purpose of this section is to derive a Vlasov equation for spin particles that is consistent with the conservation law~\eqref{consTKinetic}.
We start our argument with a vector field $W$ on the bundle $Q=\cup_{x\in M}(T_xM)^2$ of the following form:
\begin{equation}\label{W}
W=u^\beta e_{(\beta)}^\alpha X_{(\alpha)}+A^\mu U_{(\mu)}+B^\nu S_{(\nu)},
\end{equation}
where $(X_{(\alpha)},U_{(\mu)},S_{(\nu)})\psi=(\partial_{x^\alpha}\psi,\partial_{u^\mu}\psi,\partial_{s^\nu}\psi)$ for all smooth functions $\psi: Q\to\R$, $(x^\alpha, u^\mu, s^\nu)$ are the local coordinates on $Q$ introduced in Section~\ref{kineticEC} and $A^\mu, B^\mu$ are functions of $(x,u,s)$. The vector field $W$ is tangent to the state space $\Pi_{\sigma}$ if and only if, for all curves 
\[
(x^\alpha(\tau),u^\mu(\tau),s^\nu(\tau))\subset Q
\]
such that
\begin{equation}\label{integralcurves}
\frac{d x^\alpha}{d\tau}=u^\beta e_{(\beta)}^\alpha,\quad 
\frac{du^\mu}{d\tau}=A^\mu(x,u,s),\quad \frac{ds^\nu}{d\tau}= B^\nu(x,u,s),
\end{equation}
there hold
\[
\frac{d}{d\tau}\eta_{\mu\nu}u^\mu u^\nu=\frac{d}{d\tau}\eta_{\mu\nu}
s^\mu u^\nu=\frac{d}{d\tau}\eta_{\mu\nu}s^\mu s^\nu=0.
\]
Equivalently, the functions $A^\mu, B^\mu$ must satisfy
\begin{equation}\label{firstcond}
A^\mu u_\mu=0,\quad A^\mu s_\mu +B^\mu u_\mu=0,\quad B^\mu s_\mu=0.
\end{equation}
Assuming~\eqref{firstcond}, we may express $W$ in the coordinates $({\bm u}, {\bm \omega})\in\R^3\times S^2$ on $\Pi_{\sigma}[x]$ introduced in Section~\ref{kineticEC}. To achieve this, let
\begin{equation}\label{gstar}
\psi_*(x,{\bm u}, {\bm \omega})=\psi \left(x,\sqrt{1+|{\bm u}|^2}, {\bm u},\sigma\, {\bm \omega}\cdot{\bm u},\sigma\, {\bm \omega} + \frac{\sigma {\bm \omega}\cdot{\bm u}}{1+\sqrt{1+|{\bm u}|^2}}{\bm u}\right)
\end{equation}
be the restriction of $\psi=\psi(x,u,s)$ on $\{(u,s)\in\R^8\,:\,\eqref{shellnormal}\ \text{hold}\}$. 
\begin{lemma}\label{long}
For all $A^\mu,B^\mu$ satisfying~\eqref{firstcond} there holds
\begin{align*}
A^\mu(\partial_{u^\mu}\psi)_*+B^\mu(\partial_{s^\mu}\psi)_*=&\,A^i\partial_{u^i}\psi_*+\frac{(\sqrt{\mathfrak{h}})_i^j}{\sigma}\left(B^i-\frac{s^0}{1+u^0}A^i\right)\slashed{\partial}_{\omega^j}\psi_*\\
&-\frac{(\sqrt{\mathfrak{h}})_{ik}A^i\omega^k}{u^0(1+u^0)}u^j\slashed{\partial}_{\omega^j}\psi_*,
\end{align*}
where $s^0=\sigma {\bm \omega}\cdot{\bm u}$, $u^0=\sqrt{1+|{\bm u}|^2}$ and $\slashed{\partial}_{\bm \omega}=(\slashed{\partial}_{\omega^1},\slashed{\partial}_{\omega^2},\slashed{\partial}_{\omega^2})$, $\slashed{\partial}_{\omega^i}=(\delta^j_i-\omega_i\omega^j)\partial_{\omega^j}$, denotes the gradient operator on $S^2$.
\end{lemma}
\begin{proof}
See Appendix.
\end{proof}

For the purpose of computing the left hand side of~\eqref{conservationT}, it is convenient to split the functions $A^\mu, B^\mu$ as
\begin{equation}\label{ABnew}
A^\mu = \tensor{\gamma}{_\alpha_\beta^\mu}u^\alpha u^\beta+a^\mu,\quad B^\mu=\tensor{\gamma}{_\alpha_\beta^\mu}u^\alpha s^\beta+b^\mu,
\end{equation}
where $\tensor{\gamma}{_\alpha_\beta^\mu}$ are the Ricci rotation coefficients of the orthonormal frame $e_{(\mu)}$ in the Levi-Civita connection $\nabla$ and $a^\mu,b^\mu$ are arbitrary functions of $(x,{\bm u}, {\bm \omega})$ that satisfy 
\[
a^\mu u_\mu =b^\mu s_\mu=a^\mu s_\mu+b^\mu u_\mu =0,
\]
or, equivalently,
\begin{equation}\label{zerocom2}
a^0=\frac{{\bm a}\cdot{\bm u}}{u^0},\quad b^0=\frac{{\bm b}\cdot{\bm s}}{s^0},\quad (u^0{\bm s}-s^0{\bm u})\cdot(s^0{\bm a}-u^0{\bm b})=0.
\end{equation}
The identity in Lemma~\ref{long} applied to the vector fields~\eqref{ABnew} can be written in the form 
\begin{equation}\label{chain2}
(F\psi)_*=(\tensor{\gamma}{_\alpha_\beta^i}u^\alpha u^\beta +a^i)\partial_{u^i}\psi_*+\left(\tensor{\gamma}{_\alpha_k^i}u^\alpha \omega^k+y^i\right)\slashed{\partial}_{\omega^i}\psi_*
\end{equation}
for all smooth functions $\psi=\psi(x,u,s)$, where
\begin{equation}\label{U}
F=(\tensor{\gamma}{_\alpha_\beta^\mu}u^\alpha u^\beta+a^\mu)\partial_{u^\mu}+(\tensor{\gamma}{_\alpha_\beta^\mu}u^\alpha s^\beta+b^\mu)\partial_{s^\mu},
\end{equation}
and $y^i=y^i(x,{\bm u},{\bm \omega})$ are given by
\begin{equation}\label{y}
y^i=\frac{\tensor{\gamma}{_\alpha_0^j}u^\alpha}{1+u^0}({\bm \omega}\cdot{\bm u}\,\delta_j^i-\omega_j u^i)+\frac{1}{\sigma}(\sqrt{\mathfrak{h}})^i_jb^j-\frac{(\sqrt{\mathfrak{h}})_j^k}{1+u^0}\left(\frac{s^0}{\sigma}\,\delta_k^i+\frac{\omega_k u^i}{u^0}\right)a^j.
\end{equation}
As
 \[
 y^i\omega_i=\frac{1}{\sigma u^0}(\sqrt{\mathfrak{h}})_j^i(u^0b^j-s^0a^j)\omega_i=-\frac{1}{\sigma^2(u^0)^2}(u^0{\bm s}-s^0{\bm u})\cdot(s^0{\bm a}-u^0{\bm b}),
 \] 
 the third equation in~\eqref{zerocom2}  is equivalent to ${\bm y}\cdot{\bm \omega} =0$.
Thus, choosing  $a^\mu$, $b^\mu$ satisfying~\eqref{zerocom2} is equivalent to 
\begin{itemize}
\item[(i)] choosing ${\bm a}=(a^1,a^2,a^3)$ arbitrarily and setting $a^0={\bm a}\cdot{\bm u}/u^0$;
\item[(ii)] introducing ${\bm n}={\bm n}(x,{\bm u}, {\bm \omega})$, such that 
\begin{equation}\label{condb}
{\bm n}\cdot{\bm \omega}=0;
\end{equation}  
\item[(iii)] choosing $b^i$ such that $y^i= n^i$, where $y^i$ is given by~\eqref{y};
\item[(iv)] setting $b^0={\bm b}\cdot{\bm s}/s^0$.
\end{itemize}
We conclude that the most general equation $W(f_{\sigma})=0$ on the kinetic particle density  is
\begin{equation}\label{VlasovZ}
u^\alpha e_{(\alpha)}^\mu \partial_{x^\mu}f_{\sigma}+(\tensor{\gamma}{_\alpha_\beta^i}u^\alpha u^\beta +a^i)\partial_{u^i}f_{\sigma}+\left(\tensor{\gamma}{_\alpha_k^i}u^\alpha \omega^k+n^i\right)\slashed{\partial}_{\omega^i}f_{\sigma}=0,
\end{equation}
where ${\bm a},{\bm n}$ are arbitrary, up to the constraint~\eqref{condb}.

\begin{remark}\textnormal{
To express~\eqref{VlasovZ} in terms of the angles $(\theta,\varphi)$ in~\eqref{thetaphi}, one has to use the transformation
\begin{equation}\label{gradS2}
\slashed{\partial}_{{\bm \omega}} = \mathcal{U}\,\left(\begin{array}{c}\partial_\theta\\(\sin\theta)^{-1}\partial_\varphi\end{array}\right),\quad \mathcal{U}=\left(
\begin{array}{cc}
 \cos \theta  \cos \varphi  & - \sin \varphi \\
 \cos \theta  \sin \varphi  & \cos \varphi \\
 -\sin \theta  & 0 \\
\end{array}
\right),
\end{equation}
relating the gradient on $S^2$ in Cartesian and spherical coordinates.}
\end{remark}
The next step is to choose the vectors  ${\bm a},{\bm n}$ in~\eqref{VlasovZ} in such a way that the constraint~\eqref{consTKinetic} is satisfied. This step requires computing $\nabla_b (T_f)^{ab}$ from~\eqref{VlasovZ}, which we do by using the following lemma.

\begin{lemma}\label{smart}
Let $\psi=\psi(x,u,s)$ be given and $\psi_*=\psi_*(x,{\bm u}, {\bm \omega})$ be defined as in~\eqref{gstar}. Assume that $f_{\sigma}$ solves~\eqref{VlasovZ} and that the vector fields ${\bm a} f_{\sigma}$, ${\bm  n} f_{\sigma}$ are smooth. Let 
\[
\Psi^{\nu}(x)=\int_{\R^3\times S^2}\psi_*(x,{\bm u},{\bm \omega})\,u^\nu\,f_{\sigma}\,\frac{d{\bm u}\,d{\bm \omega}}{u^0}.
\]
Then 
\begin{align*}
e^{\beta}_{(\nu)}\partial_{x^\beta}\Psi^\nu = &\int_{\R^3\times S^2}[e^{\beta}_{(\nu)}u^\nu\partial_{x^\beta}\psi_*+(F\psi)_*+\tensor{\gamma}{^\beta_\alpha_\beta}u^\alpha \psi_*]\,f_{\sigma}\,\frac{d{\bm u}\,d{\bm \omega}}{u^0}\\
&+\int_{\R^3\times S^2}\left(\slashed{\partial}_{{\bm \omega}}\cdot {\bm n}+\partial_{\bm u}\cdot {\bm a}-\frac{{\bm a}\cdot{\bm u}}{(u^0)^2}\right)\psi_*f_{\sigma}\,\frac{d{\bm u}\,d{\bm \omega}}{u^0},
\end{align*}
where $F$ is given by~\eqref{U}.
\end{lemma} 
\begin{proof}
See Appendix.
\end{proof}
Applying Lemma~\ref{smart} with $\psi= u^\mu$, and using $F(u^\mu)=\tensor{\gamma}{_\alpha_\beta^\mu}u^\alpha u^\beta+a^\mu$, we find
 \[
\nabla_\nu (T_f)^{\mu\nu}=m\int_{\R^3\times S^2}\left[\left(\slashed{\partial}_{{\bm \omega}}\cdot {\bm n}+\partial_{\bm u}\cdot {\bm a}-\frac{{\bm a}\cdot{\bm u}}{(u^0)^2}\right)u^\mu +a^\mu\right]\,f_{\sigma}\,\frac{d{\bm u}\,d{\bm \omega}}{u^0}.
 \]
In order to turn the previous equation into an identity valid for all kinetic densities $f_\sigma$, we let
 \begin{equation}\label{addconst}
\left(\slashed{\partial}_{{\bm \omega}}\cdot {\bm n}+\partial_{\bm u}\cdot {\bm a}-\frac{{\bm a}\cdot{\bm u}}{(u^0)^2}\right)u^\mu +a^\mu=\frac{1}{\mathcal{M}_f}\nabla_\nu (T_f)^{\mu\nu},
 \end{equation}
 where $\mathcal{M}_f$ is the Lorentz invariant mass function
 \[
 \mathcal{M}_f=\mathcal{M}_f(x)=m\int_{\R^3\times S^2} f_{\sigma}\,\frac{d{\bm u}\,d{\bm \omega}}{u^0}.
 \]
From~\eqref{addconst} and $a^\mu u_\mu=0$ we obtain
 \begin{equation}\label{remark}
\slashed{\partial}_{{\bm \omega}}\cdot {\bm n}+\partial_{\bm u}\cdot {\bm a}-\frac{{\bm a}\cdot{\bm u}}{(u^0)^2}=-\frac{u_\mu \nabla_\nu (T_f)^{\mu\nu}}{\mathcal{M}_f}.
 \end{equation}
Substituting~\eqref{remark} in~\eqref{addconst} we find that  the constraint~\eqref{addconst} is equivalent to choosing $a^\mu$ as
 \begin{equation}\label{thetamu}
 a^\mu=\frac{1}{\mathcal{M}_f}\left(\delta_\alpha^\mu+u^\mu u_\alpha\right)\nabla_\nu (T_f)^{\alpha\nu},
 \end{equation}
 i.e., as the projection of the vector field $\mathcal{M}_f^{-1}\nabla_b (T_f)^{ab}$ onto the plane orthogonal to $u$. 
Using~\eqref{thetamu} in~\eqref{remark} gives the following equation on ${\bm n}$: 
 \begin{equation}\label{eqb}
\slashed{\partial}_{{\bm \omega}}\cdot {\bm n}=-4\frac{u_\mu\nabla_\nu (T_f)^{\mu\nu}}{\mathcal{M}_f}.
 \end{equation}
As the right hand side of~\eqref{eqb} is independent of ${\bm \omega}\in S^2$, we may write ${\bm n}$ as 
\begin{equation}\label{ndef}
{\bm n}(x,{\bm u}, {\bm \omega})=-4\frac{u_\mu\nabla_\nu (T_f)^{\mu\nu}}{\mathcal{M}_f}{\bm \xi}(x,{\bm u}, {\bm \omega}),
\end{equation}
where the vector ${\bm \xi}(x,{\bm u},{\bm \omega})$ satisfies
\begin{equation}\label{xieq}
{\bm\xi}\cdot{\bm \omega}=0,\quad \slashed{\partial}_{{\bm \omega}}\cdot {\bm \xi}=1.
\end{equation}
There exist of course infinitely many vectors ${\bm \xi}$ satisfying~\eqref{xieq}. Before presenting one example, we derive the evolution equation satisfied by the particles number current $N_f$, which is independent of the choice of ${\bm \xi}$. This equation on $N_f$ follows by
applying Lemma~\ref{smart} with $\psi_*(x,p,s)= 1$, which gives
\[
\nabla_\nu (N_f)^\nu=\int_{\R^3\times S^2}\left(\slashed{\partial}_{{\bm \omega}}\cdot {\bm n}+\partial_{\bm u}\cdot {\bm a}-\frac{{\bm a}\cdot{\bm u}}{(u^0)^2}\right)\,f_{\sigma}\,\frac{d{\bm u}\,d{\bm \omega}}{u^0}
\]
 and so, by~\eqref{remark}, 
 \begin{equation}\label{Nnotcons}
 \nabla_\mu (N_f)^\mu=-\frac{(N_f)_\mu\nabla_\nu (T_f)^{\mu\nu}}{\mathcal{M}_f}.
 \end{equation}
In conclusion, we have derived the following Vlasov equation on the kinetic density $f_\sigma=f_\sigma(x,{\bm u},{\bm \omega})$ of particles with mass $m$ and spin magnitude $\sigma$:
\begin{subequations}\label{VlasoveqParticles}
\begin{equation}
u^\alpha e_{(\alpha)}^\mu \partial_{x^\mu}f_{\sigma}+(\tensor{\gamma}{_\alpha_\beta^i}u^\alpha u^\beta +a^i)\partial_{u^i}f_{\sigma}+\left(\tensor{\gamma}{_\alpha_k^i}u^\alpha \omega^k+n^i\right)\slashed{\partial}_{\omega^i}f_{\sigma}=0,
\end{equation}
where
\begin{equation}\label{72b}
a^i=\frac{1}{\mathcal{M}_f}\left(\delta_\mu^i+u^i u_\mu\right)\nabla_\nu (T_f)^{\mu\nu},\quad n^i=-4\frac{u_\mu\nabla_\nu (T_f)^{\mu\nu}}{\mathcal{M}_f}\xi^i
\end{equation}
\end{subequations}
and ${\bm \xi}=(\xi^1,\xi^2,\xi^3)$ is an arbitrary vector that satisfies~\eqref{xieq}. We also obtained that the evolution equation for the particles number current is~\eqref{Nnotcons}. 

When only one species of particles is present, the term $\nabla_\nu (T_f)^{\mu\nu}$ in~\eqref{Nnotcons}-\eqref{VlasoveqParticles} must be replaced with $(J_f)^\nu$ in order for the constraint~\eqref{consTKinetic} to be satisfied. In particular, the Vlasov model for one species of particles derived in this section does not preserve the total number of particles. 

\subsubsection*{Example of solution to~\eqref{xieq}}
The constraint ${\bm \xi}\cdot {\bm \omega}=0$ implies that ${\bm \xi}$ 
 is the projection on the plane orthogonal to ${\bm \omega}$ of a (dimensionless) vector field ${\bm c}(x,{\bm u}, {\bm \omega})$, that is,
\begin{subequations}\label{ansatzxi}
\begin{equation}
{\bm \xi}={\bm \omega}\wedge ({\bm \omega}\wedge {\bm c}).
\end{equation}
A simple choice for the vector ${\bm c}$ is to take it parallel to the relativistic velocity ${\bm v}={\bm u}/u^0$. Specifically, we make the {\it ansatz}
\begin{equation}
{\bm c}=\zeta(z){\bm v},\quad z={\bm \omega}\cdot{\bm v}.
\end{equation}
\end{subequations}
Replacing the {\it ansatz}~\eqref{ansatzxi} into~\eqref{xieq} and computing the divergence on $S^2$ we obtain the following equation on the function $\zeta$: 
\[
\zeta'(z)(z^2-|{\bm v}|^2)+2z\zeta(z)=1,
\]
the solution of which is
\[
\zeta(z)=\frac{C-z}{|{\bm v}|^2-z^2},
\]
where $C$ is an arbitrary function of $|{\bm v}|$. Choosing $C\equiv 0$ we arrive to the following final form of the vector ${\bm \xi}$:
\[
{\bm \xi}=-\frac{{\bm \omega}\cdot{\bm v}}{|{\bm \omega}\wedge {\bm v}|^2}{\bm \omega}\wedge ({\bm \omega}\wedge {\bm v}),
\]
and therefore to the following final form of the vector ${\bm n}$:
\begin{equation}\label{finaln}
{\bm n}=4\frac{u_\mu \nabla_\nu (T_f)^{\mu\nu}}{\mathcal{M}_f}\frac{{\bm \omega}\cdot{\bm v}}{|{\bm \omega}\wedge {\bm v}|^2}{\bm \omega}\wedge ({\bm \omega}\wedge {\bm v}).
\end{equation}

\begin{remark}
\textnormal{For the choice~\eqref{finaln} of the vector ${\bm n}$, the assumption in Lemma~\ref{smart} that the vector field ${\bm n} f_{\sigma}$ should be smooth is satisfied when $f_{\sigma}=O(|{\bm n}\wedge {\bm u}|^2)$ as $|{\bm n}\wedge {\bm u}|\to 0$. }
\end{remark}
\section{The Vlasov system for two species of particles}\label{Vlasovant}
The Vlasov model derived in the previous section violates the conservation law of particles number. In this section we shall restore this fundamental property by postulating the existence in spacetime of a second species of particles interacting with the particles of the first species. We assume that the particles of the new species also have mass $m$ and spin magnitude $\sigma$. 

Let $\overline{f}_\sigma(x,{\bm u}, {\bm \omega})$ be the kinetic density of the new particles; the Vlasov equation on $\overline{f}_\sigma$ is
\begin{subequations}\label{VlasoveqAntiParticles}
\begin{equation}
u^\alpha e_{(\alpha)}^\mu \partial_{x^\mu}\overline{f}_{\sigma}+(\tensor{\gamma}{_\alpha_\beta^i}u^\alpha u^\beta +\overline{a}^i)\partial_{u^i}\overline{f}_{\sigma}+\left(\tensor{\gamma}{_\alpha_k^i}u^\alpha \omega^k+\overline{n}^i\right)\slashed{\partial}_{\omega^i}\overline{f}_{\sigma}=0,
\end{equation}
where
\begin{equation}\label{75b}
\overline{a}^i=\frac{1}{\mathcal{M}_{\overline{f}}}\left(\delta_\mu^i+u^i u_\mu\right)\nabla_\nu (T_{\overline{f}})^{\mu\nu},\quad \overline{n}^i=-4\frac{u_\mu\nabla_\nu (T_{\overline{f}})^{\mu\nu}}{\mathcal{M}_{\overline{f}}}\xi^i.
\end{equation}
\end{subequations}
The evolution equation for the number current of the new particles is therefore
  \begin{equation}\label{Nnotconsanti}
 \nabla_\mu (N_{\overline{f}})^\mu=-\frac{(N_{\overline{f}})_\mu\nabla_\nu (T_{\overline{f}})^{\mu\nu}}{\mathcal{M}_{\overline{f}}}.
 \end{equation}
As pointed out at the end of Section~\ref{matterfieldssec}, when two particle species are present it cannot be assumed that the conservation law of energy-momentum holds in the form~\eqref{consTKinetic} for each single species, since otherwise the constraint~\eqref{consmom2species} would be violated.
To overcome this inconsistency, we postulate that the stress-energy tensors of the two particle species satisfy 
\begin{equation}\label{smart1}
\nabla_b (T_f)^{ab}=\chi(x)(J_{f+\overline{f}})^a,\quad \nabla_b (T_{\overline{f}})^{ab}=\overline{\chi}(x)(J_{f+\overline{f}})^a,
\end{equation}
for some functions $\chi,\overline{\chi}$ such that
\begin{equation}\label{alpha1}
\chi(x)+\overline{\chi}(x)=1.
\end{equation}
Replacing~\eqref{smart1} in~\eqref{Nnotcons} and~\eqref{Nnotconsanti} we find the following equation on the total particles number current of the two species:
\[
 \nabla_\mu [  (N_f)^\mu+(N_{\overline{f}})^\mu]=-\chi(x)\frac{(N_f)_\mu(J_{f+\overline{f}})^\mu}{\mathcal{M}_f}-\overline{\chi}(x)\frac{(N_{\overline{f}})_\mu(J_{f+\overline{f}})^\mu}{\mathcal{M}_{\overline{f}}}.
\]
Hence, we obtain that $ (N_f)^\mu+(N_{\overline{f}})^\mu$ is divergence-free when
\begin{equation}\label{alpha2}
\chi(x)\frac{(N_f)_\mu(J_{f+\overline{f}})^\mu}{\mathcal{M}_f}+\overline{\chi}(x)\frac{(N_{\overline{f}})_\mu(J_{f+\overline{f}})^\mu}{\mathcal{M}_{\overline{f}}}=0.
\end{equation}
Solving the system~\eqref{alpha1}-\eqref{alpha2} and replacing the solution $\chi,\overline{\chi}$ in~\eqref{smart1} we obtain
\begin{subequations}\label{smart2}
\begin{align}
&\nabla_b (T_f)^{ab}=-\frac{(N_{\overline{f}})_\mu(J_{f+\overline{f}})^\mu \mathcal{M}_f}{(J_{f+\overline{f}})^\mu(\mathcal{M}_{\overline{f}}(N_f)_\mu-\mathcal{M}_f(N_{\overline{f}})_\mu)}(J_{f+\overline{f}})^a,\\
& \nabla_b (T_{\overline{f}})^{ab}=\frac{(N_{f})_\mu(J_{f+\overline{f}})^\mu \mathcal{M}_{\overline{f}}}{(J_{f+\overline{f}})^\mu(\mathcal{M}_{\overline{f}}(N_f)_\mu-\mathcal{M}_f(N_{\overline{f}})_\mu)}(J_{f+\overline{f}})^a.
\end{align}
\end{subequations}
Thus, our final system on the kinetic densities $f_\sigma, \overline{f}_\sigma$ for the particles of the two species is~\eqref{VlasoveqParticles}-\eqref{VlasoveqAntiParticles} with~\eqref{smart2} replaced in the definitions of the vectors ${\bm a}, \overline{\bm a}, {\bm n}, \overline{\bm n}$. This system couples the dynamics of the kinetic densities $f_\sigma,\overline{f}_\sigma$ and is consistent with the conservation law of energy-momentum~\eqref{consmom2species}. The total number of particles of each individual species is not conserved, but their sum is. 

\section{Conclusions}\label{conclu}
In this paper we have introduced a new general relativistic kinetic model for the dynamics of spin neutral particles with positive mass in a spacetime with torsion. The main assumptions of the model are
\begin{itemize}
\item[(i)] the four-velocity $u$ and four-spin $s$ of the particles are constrained by $u^\mu s_\mu =0$ and $s^\mu s_\mu =\sigma^2$ for a positive constant $\sigma\neq 1$; 
\item[(ii)] the particles do not collide; 
\item[(iii)] the particles spin induces a torsion in spacetime that obeys Eintein-Cartan's theory.  
\end{itemize}
We have derived the most general transport equation on the kinetic particle density $f(x,u,s)$ that is consistent with assumptions (i)-(ii); see~\eqref{VlasovZ}. This equation is defined up to the choice of two arbitrary vectors, which can be chosen so that the model is compatible with the Bianchi identity (i.e., the conservation law of energy-momentum) in Einstein-Cartan's theory. The total number of particles is not conserved by the single species particle model, which led us to assume the existence in spacetime of an additional species of particles with the same mass $m$ and spin magnitude $\sigma$. The evolution equation for the kinetic density $\overline{f}(x,u,s)$  of these new particles has been obtained by imposing that the total particles number current computed with the kinetic density $f+\overline{f}$ should be divergence-free. 

According to Equations~\eqref{VlasoveqParticles}-\eqref{VlasoveqAntiParticles}, the particles motion is not geodesic. In fact, along their trajectory $x=x(\tau)$, the momentum ${\bm p}(\tau)=m{\bm u}(\tau)$ of the two particle species satisfies, respectively,
\[
\frac{dp^i}{d\tau}=m(\tensor{\gamma}{_\alpha_\beta^i}u^\alpha u^\beta +a^i),
\quad
\frac{dp^i}{d\tau}=m(\tensor{\gamma}{_\alpha_\beta^i}u^\alpha u^\beta +\overline{a}^i),
\]
where $a^i, \overline{a}^i$ are given by~\eqref{72b},~\eqref{75b}. Thus, the deviation from geodesics motion is different for the particles of the two species: in one case it is determined by the force $m{\bm a}$, in the second case by the force $m\overline{{\bm a}}$. Both these forces are induced by the spacetime torsion. 
Similarly, the direction ${\bm \omega}$ of the rest frame spin vector for the two particle species obeys, respectively,
\[
\frac{d\omega^i}{d\tau}=\tensor{\gamma}{_\alpha_k^i}u^\alpha \omega^k+n^i,\quad \frac{d\omega^i}{d\tau}=\tensor{\gamma}{_\alpha_k^i}u^\alpha \omega^k+\overline{n}^i,
\]
where $n^i,\overline{n}^i$ are given in~\eqref{72b},~\eqref{75b}.
The first term in the right hand side of the previous equations is the change of ${\bm \omega}$ due to the rotation of the frame $e_{(\mu)}$, while the terms $n^i,\overline{n}^i$ represent the rotation of ${\bm \omega}$ induced by the spacetime torsion. Again, this effect is different for the two particle species.

In the typical applications of the Vlasov equation~\eqref{XVlasov} for spinless particles, $f_*$ stands for the kinetic density of galaxies, or even clusters of galaxies~\cite{BT}. In our model, the kinetic densities $f_\sigma$, $\overline{f}_\sigma$ are to be interpreted as the densities in state space of elementary particles. In this respect, the problem to which our model could be applicable is the study of the early universe dynamics.  
   
Since we assume that particles are neutral and have positive mass, then our model applies, for instance, to neutrinos. Within this application it is tempting to identify the two particle species with neutrinos/antineutrinos pairs,  but this interpretation is not without issues. In fact, according to the standard model neutrinos and antineutrinos are distinguished based on their weak isospin, and thus on their weak interaction, while in our model they differ by how torsion acts on them.
Moreover, the difference between the action of gravity on neutrinos and antineutrinos predicted by our model has never been observed.

Finally, we remark that the results presented in this paper are entirely different from those in~\cite{Galia}, which pertain to exact solutions of the standard Vlasov equation for particles without spin on a background symmetric manifold with a given simple torsion.

{\bf Acknowledgments.} I am grateful to H\aa kan Andr\'easson for his comments on this article.
\appendix 
\setcounter{secnumdepth}{0}
\section{Appendix}
\subsection*{Proof of Lemma~\ref{long}}
We have
\begin{align}
&\partial_{u^i}\psi_*=\frac{u_i}{u^0}\,(\partial_{u^0}\psi)_* +(\partial_{u^i}\psi)_* +\sigma \omega_i(\partial_{s^0}\psi)_* +\frac{\sigma \omega_iu^j}{1+u^0}(\partial_{s^j}\psi)_*+\frac{\sigma{\bm \omega}\cdot{\bm u}}{1+u^0}(\sqrt{\mathfrak{h}})_i^{j}(\partial_{s^j}\psi)_*,\label{prima}\\
&\partial_{\omega^i}\psi_*=\sigma u_i(\partial_{s^0}\psi)_* +\sigma [(\sqrt{\mathfrak{h}})^{-1}]_{i}^{j}(\partial_{s^j}\psi)_* .\label{seconda}
\end{align}
Inverting~\eqref{seconda} we obtain
\begin{equation}\label{tempo}
(\partial_{s^j}\psi)_*=-\frac{u_j}{u^0}(\partial_{s^0}\psi)_*+\frac{1}{\sigma}(\sqrt{\mathfrak{h}})^{k}_{j}\partial_{\omega^k}\psi_*,
\end{equation}
where we used that $(\sqrt{\mathfrak{h}})^{i}_{j}u_i=u_j/u^0$.
By~\eqref{tempo} we have
\begin{equation}\label{pmusu}
u^j(\partial_{s^j}\psi)_*=-\frac{|{\bm u}|^2}{u^0}(\partial_{s^0}\psi)_*+\frac{u^k}{\sigma u^0}\partial_{\omega^k}\psi_*.
\end{equation} 
Using~\eqref{tempo} and~\eqref{pmusu} in~\eqref{prima} we find, after some simplifications,
\begin{align}
\partial_{u^i}\psi_*&=\frac{u_i}{u^0}(\partial_{u^0}\psi)_*+(\partial_{u^i}\psi)_*+\frac{s^0}{1+u^0}(\partial_{s^i}\psi)_*+\frac{s_i}{u^0}(\partial_{s^0}\psi)_*-\frac{s^0u_i}{(u^0)^2(1+u^0)}(\partial_{s^0}\psi)_*\nonumber\\
&\quad +\frac{(\sqrt{\mathfrak{h}})^k_i \omega_k}{u^0(1+u^0)}u^j\partial_{\omega^j}\psi_*.\label{dpig2}
\end{align}
Hence, using $A^\mu u_\mu =A^i u_i-A^0u^0=0$, we obtain
\begin{align}
A^i\partial_{u^i}\psi_*&=A^\mu(\partial_{u^\mu}\psi)_*+\frac{A^\mu s_\mu}{u^0}(\partial_{s^0}\psi)_*+\frac{s^0}{1+u^0}A^\mu(\partial_{s^\mu}\psi)_*\nonumber\\
&\quad+\frac{(\sqrt{\mathfrak{h}})_{ik}A^i \omega^k}{u^0(1+u^0)}u^j\partial_{\omega^j}\psi_*.\label{almostthere}
\end{align}
Furthermore, again by~\eqref{tempo},
\begin{subequations}\label{uau}
\begin{align}
&A^\mu(\partial_{s^\mu}\psi)_*=\frac{1}{\sigma}(\sqrt{\mathfrak{h}})_i^j A^i\partial_{\omega^j}\psi_*,\\
& B^\mu(\partial_{s^\mu}\psi)_*=\frac{A^\mu s_\mu}{u^0}(\partial_{s^0}\psi)_*+\frac{1}{\sigma}(\sqrt{\mathfrak{h}})_i^j B^i\partial_{\omega^j}\psi_*,
\end{align}
\end{subequations}
where for the second identity we used $A^\mu s_\mu+B^\mu u_\mu=0$. 
Combining~\eqref{almostthere} and~\eqref{uau} we find
\begin{align}
A^\mu(\partial_{u^\mu}\psi)_*+B^\mu(\partial_{s^\mu}\psi)_*=&A^i\partial_{u^i}\psi_*+\frac{(\sqrt{\mathfrak{h}})_i^j}{\sigma}\left(B^i-\frac{s^0}{1+u^0}A^i\right)\partial_{\omega^j}\psi_*\nonumber\\
&-\frac{(\sqrt{\mathfrak{h}})_{ik}A^i\omega^k}{u^0(1+u^0)}u^j\partial_{\omega^j}\psi_*.\label{temporalicchio}
\end{align}
Now, in the right hand side of~\eqref{temporalicchio} we replace $\partial_{\omega^j}=\slashed{\partial}_{\omega^j}+\omega_j\omega^k\partial_{\omega^k}$ to get
\begin{align*}
RHS~\eqref{temporalicchio}=&A^i\partial_{u^i}\psi_*+\frac{(\sqrt{\mathfrak{h}})_i^j}{\sigma}\left(B^i-\frac{s^0}{1+u^0}A^i\right)\slashed{\partial}_{\omega^j}\psi_*-\frac{(\sqrt{\mathfrak{h}})_{ik}A^i\omega^k}{u^0(1+u^0)}u^j\slashed{\partial}_{\omega^j}\psi_*\\
&+\frac{(\sqrt{\mathfrak{h}})_i^j}{u^0\sigma}\left(u^0B^i-s^0A^i\right)\omega_j\omega^k\partial_{\omega^k}\psi_*.
\end{align*}
Since
\[
(\sqrt{\mathfrak{h}})_i^j(u^0B^i-s^0A^i)\omega_j=-\frac{s^0}{\sigma}\,(A^\mu v_\mu+B^\mu u_\mu)=0,
\]
the proof is complete.

\subsection*{Proof of Lemma~\ref{smart}}
We have
\[
e_{(\nu)}^\beta\partial_{x^\beta}\Psi^\nu=\int_{\R^3\times S^2}f_{\sigma}\,u^\nu e_{(\nu)}^\beta \partial_{x^\beta}\psi_*\frac{d{\bm u}\,d{\bm \omega}}{u^0}+\int_{\R^3\times S^2}\psi_* u^\nu e_{(\nu)}^\beta\partial_{x^\beta}f_{\sigma}\,\frac{d{\bm u}\,d{\bm \omega}}{u^0}=I+II.
\]
By~\eqref{VlasovZ} the second integral is
\begin{align*}
II&=-\int_{\R^3\times S^2}\psi_*\left((\tensor{\gamma}{_\alpha_\beta^i}u^\alpha u^\beta+a^i) \partial_{u^i}f_{\sigma}+\left(\tensor{\gamma}{_\alpha_k^i}u^\alpha \omega^k+n^i\right)\slashed{\partial}_{\omega^i}f_{\sigma}\right)\frac{d{\bm u}\,d{\bm \omega}}{u^0}\\
&=\int_{\R^3\times S^2}((\tensor{\gamma}{_\alpha_\beta^i}u^\alpha u^\beta+a^i)\partial_{u^i}\psi_*+\left(\tensor{\gamma}{_\alpha_k^i}u^\alpha \omega^k+n^i\right)\slashed{\partial}_{\omega^i}\psi_*+\tensor{\gamma}{^\beta_\alpha_\beta}u^\alpha \psi_*)\,f_{\sigma}\,\frac{d{\bm u}\,d{\bm \omega}}{u^0}\\
&\quad\quad+ \int_{\R^3\times S^2}\psi_*(\slashed{\partial}_{\omega^i}n^i+\partial_{u^i}a^i-a^iu_i/(u^0)^2)\,f_{\sigma}\,\frac{d{\bm u}\,d{\bm \omega}}{u^0},
\end{align*}
where we integrated by parts in the variables ${\bm u}, {\bm \omega}$ and used that
\[
\tensor{\gamma}{_\alpha_\beta^i}\partial_{u^i}\left(\frac{u^\alpha u^\beta}{u^0}\right)=\tensor{\gamma}{^\beta_\alpha_\beta}\frac{u^\alpha}{u^0},\quad \slashed{\partial}_{\omega^i}(\tensor{\gamma}{_\alpha_k^i}u^\alpha \omega^k)=0.
\] 
(We also make the tacit assumption that $f_{\sigma}$ decays to zero as $|{\bm u}|\to \infty$ sufficiently fast so that the boundary terms arising from the integration by parts in the ${\bm u}$ variable vanish.)
The proof now follows from~\eqref{chain2} with $y^i=n^i$.

\end{document}